\documentclass[12pt,a4paper]{article}
\usepackage{amsmath,amssymb,amsthm}
\theoremstyle{plain}
\newtheorem{thm}{Theorem}[section]
\newtheorem{prop}[thm]{Proposition}
\newtheorem{cor}[thm]{Corollary}
\newtheorem{lem}[thm]{Lemma}
\newtheorem{conj}[thm]{Conjecture}

\theoremstyle{definition}
\newtheorem{defn}[thm]{Definition}

\newtheorem{rem}[thm]{Remark}
\newtheorem{prob}[thm]{Problem}

\newcommand{\NN}{{\mathbb{N}}}
\newcommand{\QQ}{{\mathbb{Q}}}
\newcommand{\RR}{{\mathbb{R}}}
\newcommand{\ZZ}{{\mathbb{Z}}}

\newcommand{\Ac}{{\mathcal{A}}}

\newcommand{\qqed}{\hbox{\rule{6pt}{6pt}}}
\newcommand{\F}{\mathcal{F}}
\newcommand{\B}{\mathcal{B}}
\newcommand{\C}{\mathcal{C}}
\newcommand{\R}{\mathbb{R}}
\usepackage{float}
\begin{document}

\title{Feedback game on $3$-chromatic Eulerian triangulations of surfaces}
\author{Akihiro Higashitani\thanks{Department of Pure and Applied Mathematics, Graduate School of Information Science and Technology, Osaka University, Suita, Osaka 565-0871, Japan.
E-mail:{\tt higashitani@ist.osaka-u.ac.jp}},
Kazuki Kurimoto\thanks{
Department of Mathematics, Graduate School of Science, Kyoto Sangyo University, Kyoto 603-8555, Japan.
E-mail:{\tt i1885045@cc.kyoto-su.ac.jp}
}
and 
Naoki Matsumoto\thanks{Research Institute for Digital Media and Content, Keio University, Yokohama, Kanagawa 232-0062, Japan.
E-mail:{\tt naoki.matsumo10@gmail.com}}
}
\date{}
\maketitle

\begin{abstract}
Recently, a new impartial game on a connected graph has been introduced, called a \textit{feedback game}, which is a variant of the generalized geography. 
In this paper, we study the feedback game on $3$-chromatic Eulerian triangulations of surfaces. 
We prove that the winner of the game on every $3$-chromatic Eulerian triangulation of a surface all of whose vertices have degree $0$ modulo $4$ is always fixed. 
Moreover, we also study the case of $3$-chromatic Eulerian triangulations of surfaces which have at least two vertices whose degrees are $2$ modulo $4$. 
In addition, as a concrete class of such graphs, we consider the octahedral path, which is obtained from an octahedron by adding octahedra in the same face,
and determine the winner of the game on those graphs. 
\end{abstract}



\section{Introduction}
All graphs considered in this paper are finite simple undirected graphs. 
We say that a graph $G$ is \textit{Eulerian} if every vertex of $G$ has even degree. 
We refer the reader to \cite{D} for the basic terminologies. 

Recently, a new impartial game on a connected graph has been introduced, called a \textit{feedback game}.  
\begin{defn}[{\cite{MN}}]
There are two players, Alice and Bob. Alice starts the game. 
For a given connected graph $G$ with a starting vertex $s$, a token is put on $s$.
They alternately move the token on a vertex $u$ to another vertex $v$ which is adjacent to $u$ and the edge $uv$ will be deleted after he or she moves the token. 
The first player who is able to move the token back to the starting vertex $s$ or 
to an isolated vertex (after removing an edge used by the last move) wins the game. 
\end{defn}

The feedback game is a variant of the \textit{generalized geography}, 
which is a most famous two-player impartial game played on (directed) graphs.
It is known that for many variants of the game 
to determine the winner of the game is PSPACE-complete (see~\cite{FSU,LS,S} for example).
Similarly, the decision problem of the feedback game is PSPACE-complete
even if the maximal degree of a given connected graph is~3 
or it is~4 and the graph is Eulerian~\cite{MN}.
Thus it is an important problem
to determine the winner of the game on concrete classes of graphs.

In the recent paper~\cite{MN}, they focus on connected Eulerian graphs,
since the token is finally moved back to the starting vertex if a given connected graph is Eulerian.
It is trivial that Bob wins the game if a given connected Eulerian graph is bipartite, and hence, they deal with non-bipartite ones in that paper.
In particular, 
they gave sufficient conditions for Bob to win the game on two concrete classes of such graphs, triangular grid graphs and toroidal grid graphs,
and they also conjectured that those conditions are necessary.

One of challenging problems in this study is to determine the winner of the game on $3$-chromatic graphs.
In this paper, we focus on the feedback game on $3$-chromatic Eulerian triangulations of surfaces 
and determine the winner of the game,
where a \textit{triangulation} is a graph of a surface with each face triangular and every two faces sharing at most one edge.
(In this paper, a \textit{surface} means a closed surface, i.e., a 2-dimensional manifold without boundary.)
Let $c_{n,k}(G)$ be the number of vertices of a graph $G$ whose degrees are congruent to $k$ modulo $n$; note that $G$ satisfies $|V(G)|=c_{4,0}(G)+c_{4,2}(G)$ if $G$ is Eulerian. 

We begin with the following general result and a remark for the game on Eulerian triangulations.
\begin{thm}\label{thm:deg4}
Let $G$ be a $3$-chromatic Eulerian triangulation of a surface. 
If $c_{4,2}(G)=0$, then Bob always wins the game on $G$ with any starting vertex.
\end{thm}
\begin{rem}\label{rem:deg4}
It is well known that every Eulerian triangulation of the sphere is $3$-chromatic. 
Moreover, we can see by the Handshaking Lemma that 
we have $c_{4,2}(G) \geq 2$ for every $3$-chromatic Eulerian triangulation $G$ of a surface with $c_{4,2}(G) \neq 0$
and that if $c_{4,2}(G) = 2$, then the two such vertices belong to the same partite set of $G$.
\end{rem}

By the above theorem and remark,
we mainly deal with Eulerian triangulations $G$ of surfaces with $c_{4,2}(G) \geq 2$.
It is relatively easy to see that Bob wins the game on every Eulerian \textit{double wheel}, 
which is an Eulerian triangulation on the sphere with $n \geq 6$ vertices and degree sequence $(4,4,\dots,4,n-2,n-2)$,
for any starting vertex (Proposition~\ref{prop:dw}).
On the other hand, we find Eulerian triangulations on which Alice wins the game, as follows.

\begin{thm}\label{thm:c_{4,2}=m}
For any surface $F^2$ and any positive integer $m \geq 2$, 
there exists a $3$-chromatic Eulerian triangulation $G$ of $F^2$ 
with $c_{4,2}(G)=m$ and a starting vertex such that Alice wins the game on $G$. 
\end{thm}

Thus we study one of reasonable concrete classes of Eulerian triangulations,
called an \textit{octahedral path}, denoted by $E_n$. (See Figure~\ref{fig1}.) This graph is obtained from the octahedron (i.e., $E_1$) 
by repeatedly adding a new octahedron in the same face, which way looks like a path. 
By the construction,
the degree sequence of $E_n$ is $(4,4,4,4,4,4, \underbrace{6,6,\dots,6}_{3(n-1)})$.
Note that every Eulerian triangulation of the sphere can be constructed from $E_1$ by two operations, one of which is the octahedron addition. 
(The \textit{octahedron addition} is to add three vertices $a_1,a_2,a_3$ to a face $u_1,u_2,u_3$ of an Eulerian triangulation 
and edges so that $a_iu_ju_k$ and $a_ia_ju_k$ are faces of the resulting graph for $\{i,j,k\} = \{1,2,3\}$. See~\cite[Theorem 3]{B} for more detail.)
\begin{figure}[htb]
\centering
\unitlength 0.1in
\begin{picture}(39.6500, 28.4000)(2.8500,-32.5000)
%
\special{pn 8}%
\special{sh 1.000}%
\special{ar 2400 600 50 50  0.0000000 6.2831853}%
%
\special{pn 8}%
\special{sh 1.000}%
\special{ar 600 3200 50 50  0.0000000 6.2831853}%
%
\special{pn 8}%
\special{sh 1.000}%
\special{ar 4200 3200 50 50  0.0000000 6.2831853}%
%
\special{pn 13}%
\special{pa 2400 600}%
\special{pa 600 3200}%
\special{fp}%
%
\special{pn 13}%
\special{pa 600 3200}%
\special{pa 4200 3200}%
\special{fp}%
%
\special{pn 13}%
\special{pa 4200 3200}%
\special{pa 2400 600}%
\special{fp}%
%
\special{pn 8}%
\special{pa 2400 600}%
\special{pa 2400 1400}%
\special{fp}%
%
\special{pn 8}%
\special{pa 1400 2800}%
\special{pa 600 3200}%
\special{fp}%
%
\special{pn 8}%
\special{pa 3400 2800}%
\special{pa 4200 3200}%
\special{fp}%
%
\special{pn 8}%
\special{pa 2400 2000}%
\special{pa 2200 2400}%
\special{fp}%
%
\special{pn 8}%
\special{pa 2200 2400}%
\special{pa 2600 2400}%
\special{fp}%
%
\special{pn 8}%
\special{pa 2600 2400}%
\special{pa 2400 2000}%
\special{fp}%
%
\special{pn 8}%
\special{pa 2400 2000}%
\special{pa 2400 1800}%
\special{fp}%
%
\special{pn 8}%
\special{pa 2400 600}%
\special{pa 3800 3000}%
\special{fp}%
%
\special{pn 8}%
\special{pa 4200 3200}%
\special{pa 1000 3000}%
\special{fp}%
%
\special{pn 8}%
\special{pa 600 3200}%
\special{pa 2400 1000}%
\special{fp}%
%
\special{pn 13}%
\special{pa 2400 1000}%
\special{pa 3800 3000}%
\special{fp}%
%
\special{pn 13}%
\special{pa 3800 3000}%
\special{pa 1000 3000}%
\special{fp}%
%
\special{pn 13}%
\special{pa 1000 3000}%
\special{pa 2400 1000}%
\special{fp}%
%
\special{pn 8}%
\special{pa 1000 3000}%
\special{pa 2400 1400}%
\special{fp}%
%
\special{pn 8}%
\special{pa 2400 1600}%
\special{pa 2400 1400}%
\special{fp}%
%
\special{pn 8}%
\special{pa 2200 2400}%
\special{pa 2000 2500}%
\special{fp}%
%
\special{pn 8}%
\special{pa 2800 2500}%
\special{pa 2600 2400}%
\special{fp}%
%
\special{pn 8}%
\special{sh 1.000}%
\special{ar 2400 1000 50 50  0.0000000 6.2831853}%
%
\special{pn 8}%
\special{sh 1.000}%
\special{ar 2400 1400 50 50  0.0000000 6.2831853}%
%
\special{pn 8}%
\special{sh 1.000}%
\special{ar 2400 2000 50 50  0.0000000 6.2831853}%
%
\special{pn 8}%
\special{sh 1.000}%
\special{ar 2200 2400 50 50  0.0000000 6.2831853}%
%
\special{pn 8}%
\special{sh 1.000}%
\special{ar 2600 2400 50 50  0.0000000 6.2831853}%
%
\special{pn 8}%
\special{sh 1.000}%
\special{ar 3800 3000 50 50  0.0000000 6.2831853}%
%
\special{pn 8}%
\special{sh 1.000}%
\special{ar 1000 3000 50 50  0.0000000 6.2831853}%
%
\special{pn 8}%
\special{sh 1.000}%
\special{ar 1400 2800 50 50  0.0000000 6.2831853}%
%
\special{pn 8}%
\special{sh 1.000}%
\special{ar 3400 2800 50 50  0.0000000 6.2831853}%
%
\special{pn 8}%
\special{pa 3400 2800}%
\special{pa 3100 2650}%
\special{fp}%
%
\special{pn 8}%
\special{pa 1700 2650}%
\special{pa 1400 2800}%
\special{fp}%
%
\special{pn 8}%
\special{pa 3800 3000}%
\special{pa 1400 2800}%
\special{fp}%
%
\special{pn 13}%
\special{pa 1400 2800}%
\special{pa 3400 2800}%
\special{fp}%
%
\special{pn 13}%
\special{pa 3400 2800}%
\special{pa 2400 1400}%
\special{fp}%
%
\special{pn 13}%
\special{pa 2400 1400}%
\special{pa 1400 2800}%
\special{fp}%
%
\special{pn 8}%
\special{pa 2400 1000}%
\special{pa 3400 2800}%
\special{fp}%
%
\special{pn 8}%
\special{pa 3400 2800}%
\special{pa 2650 2710}%
\special{da 0.070}%
%
\special{pn 8}%
\special{pa 1420 2790}%
\special{pa 1820 2390}%
\special{da 0.070}%
\special{pa 1820 2390}%
\special{pa 1820 2390}%
\special{da 0.070}%
%
\special{pn 8}%
\special{pa 2400 1400}%
\special{pa 2600 1800}%
\special{da 0.070}%
%
\special{pn 8}%
\special{pa 2400 1600}%
\special{pa 2400 1800}%
\special{dt 0.045}%
%
\special{pn 8}%
\special{pa 2000 2490}%
\special{pa 1700 2650}%
\special{dt 0.045}%
%
\special{pn 8}%
\special{pa 3100 2650}%
\special{pa 2800 2500}%
\special{dt 0.045}%
%
\special{pn 8}%
\special{pa 2400 2000}%
\special{pa 2200 2200}%
\special{da 0.070}%
%
\special{pn 8}%
\special{pa 2200 2400}%
\special{pa 2410 2500}%
\special{da 0.070}%
%
\special{pn 8}%
\special{pa 2610 2400}%
\special{pa 2580 2140}%
\special{da 0.070}%
\put(24.8500,-4.9500){\makebox(0,0){$u_0$}}%
\put(25.0000,-9.2000){\makebox(0,0){$u_1$}}%
\put(24.9000,-13.1000){\makebox(0,0){$u_2$}}%
\put(25.2000,-19.1000){\makebox(0,0){$u_n$}}%
\put(5.1000,-33.0000){\makebox(0,0){$v_0$}}%
\put(43.3000,-32.9000){\makebox(0,0){$w_0$}}%
\put(10.9000,-31.0000){\makebox(0,0){$v_1$}}%
\put(14.9000,-29.0000){\makebox(0,0){$v_2$}}%
\put(22.7000,-25.3000){\makebox(0,0){$v_n$}}%
\put(27.2000,-23.2000){\makebox(0,0){$w_n$}}%
\put(34.7000,-27.1000){\makebox(0,0){$w_2$}}%
\put(38.9000,-29.0000){\makebox(0,0){$w_1$}}%
\end{picture}
\caption{Octahedral path $E_n$}
\label{fig1}
\end{figure}
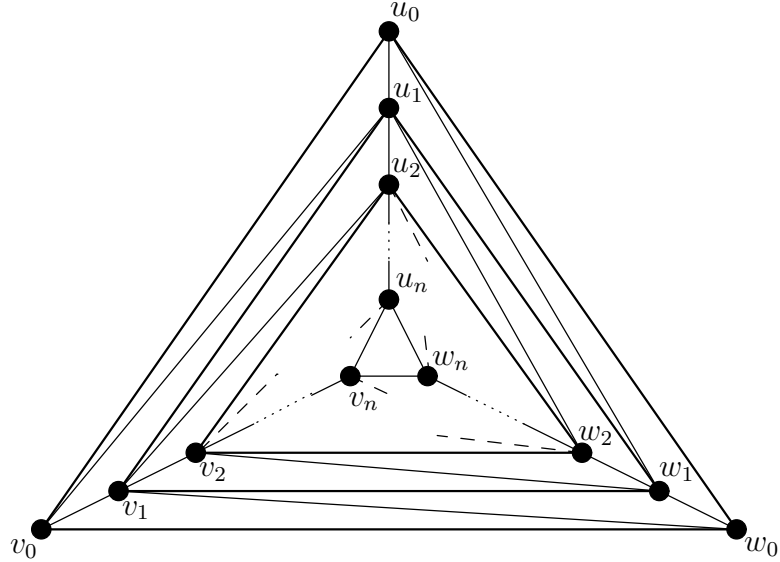

The feedback game on the graph $E_n$ with its starting vertex $u_p$ or $v_p$ or $w_p$ is denoted by $E(n,p)$, where $n\geq 1$ and $0 \leq p \leq n$. 
Notice that the choice of $u_p,v_p,w_p$ does not matter. 
We give the complete characterization whether Alice or Bob wins the game on $E(n,p)$:
\begin{thm}\label{thm:octapath}
For any integers $n$ and $p$ with $n\geq 1$ and $0 \leq p \leq n$, 
the winner of the game $E(n,p)$ is determined as follows: 
\begin{center}
\begin{tabular}{c|c|c|c} 
             &$n \equiv 0$ &$n \equiv 1$ &$n \equiv 2$ \\ \hline
$p \equiv 0$ &Alice &Bob &Alice \\ \hline
$p \equiv 1$ &Alice &Bob &Bob \\ \hline
$p \equiv 2$ &Alice &Alice &Alice \\ 
\end{tabular}
\end{center}
Here, all ``$\equiv$'' mean that ``congruent modulo $3$''.
\end{thm}

In the next section, we introduce an important concept for the feedback game, 
called an \textit{even kernel}\/, 
and we prove Theorems~\ref{thm:deg4} and~\ref{thm:c_{4,2}=m} and Proposition~\ref{prop:dw}.
In Section~\ref{sec:octa},
we shall give a proof of Theorem~\ref{thm:octapath}.


\section{$3$-chromatic Eulerian triangulations}\label{matsumoto}

We first introduce the key notion, called an \textit{even kernel}\/, which is first introduced in~\cite{FSU}.

\begin{defn}[Even kernel~\cite{FSU}]
Let $G$ be a connected graph with a starting vertex $s$. 
An \textit{even kernel} for $G$ is a nonempty subset $S \subset V(G)$ such that 
\begin{itemize}
\item $s \in S$,
\item no two vertices in $S$ are adjacent, and 
\item every vertex not in $S$ is adjacent to an even number (possibly $0$) of vertices in $S$. 
\end{itemize}
\end{defn}

\begin{defn}[{Even kernel graph~\cite{MN}}]
Let $S$ be an even kernel of a connected Eulerian graph $G$ with a starting vertex. 
An \textit{even kernel graph} with respect to $S$ is 
a bipartite subgraph $H_S$ with the bipartition $V(H_S)=S \cup R$ 
and $E(H_S)=E_G(S,R)$, 
where $R$ is a superset of the set 
$N_G(S) = \{v \in V(G) \setminus S : v $ is adjacent to a vertex $u \in S\}$
and $E_G(A,B)$ denotes the set of edges between $A$ and $B$. 
\end{defn}

\begin{rem}\label{rem:even-kernel}
It is easy to see that if a connected graph $G$ with some starting vertex 
has an even kernel, then Bob wins the feedback game on $G$ 
since Bob can always move a token from a vertex not in $S$ back to a vertex in $S$ \cite{FSU, MN}. 
As a corollary, Bob also wins the game on every connected Eulerian bipartite graph.
It is NP-complete in general to find an even kernel of a given graph~\cite{F}.
Furthermore, 
there exist infinitely many connected Eulerian graphs without an even kernel (for a specified starting vertex)
on which Bob wins the game~\cite{MN}.
\end{rem}

Now we shall prove Theorems~\ref{thm:deg4} and~\ref{thm:c_{4,2}=m}.

\begin{proof}[Proof of Theorem~\ref{thm:deg4}]
Let $G$ be a $3$-chromatic Eulerian triangulation on a surface and let $(S_1,S_2,S_3)$ be the tripartite set of $G$, 
i.e., $V(G) = \bigcup^3_{i=1}S_i$ and $S_i \cap S_j = \emptyset$ for $i\neq j$. 
Without loss of generality, the starting vertex $s$ is in $S_1$. 
Since the degree of each vertex of $G$ is zero modulo~$4$, 
each vertex in $S_i$ is adjacent to an even number of vertices in $S_j$ for any $i,j \in \{1,2,3\}$ with $i \neq j$. 
(Note that the subgraph induced by the neighbors of each vertex in $G$ is an even cycle of length $4k$ for some $k\geq 1$.) 
Therefore, $S_1$ is clearly an even kernel. 
\end{proof}

\begin{proof}[Proof of Theorem~\ref{thm:c_{4,2}=m}] 
\noindent
{\bf (The first step)}: 
It suffices to prove the theorem for the sphere, since we can construct the desired Eulerian triangulation on a given surface $F^2$
by ``pasting" a face of the desired one on the sphere (constructed below)
and that of any $3$-chromatic Eulerian triangulation on $F^2$ 
with only vertices of degree~$4k$ with $k \geq 1$. 

\begin{rem}
Let $G$ be a $3$-chromatic Eulerian triangulation on $F^2$ and let $uvw$ and $uvw'$ be two faces of $G$.
A \textit{2-subdivision} of an edge $uv$ of $G$ is to replace $uv$ with a path $uabv$ and to add edges $aw,bw,aw',bw'$.
Note that $w$ and $w'$ belong to the same partite set of $G$
and that the remainder of their degrees modulo~4 is changed after this operation.
The octahedron addition also changes the remainder modulo~$4$
of degrees of three vertices on the corresponding face boundary.
Thus applying octahedron additions and 2-subdivisions suitably, 
we can construct a $3$-chromatic Eulerian triangulation on $F^2$ with only vertices of degree~$4k$ with $k \geq 1$ from another by Remark~\ref{rem:deg4}. 
\end{rem}

Furthermore, for any $m \equiv 0 \pmod{3}$ with $m\geq 3$, the octahedral path $E_n$ with $n \geq 2$ is the desired graph. See Theorem~\ref{thm:octapath}. 

Hence, we shall construct the desired graph for any $m \geq 2$ with $m \equiv 1,2 \pmod{3}$ in the second and third steps. 

\noindent 
{\bf (The second step)}: 
Let $G$ be the Eulerian triangulation on the sphere shown in the left of Figure~\ref{basetri}.
Note that exactly two vertices $u_1$ and $b$ are of degree~$6$
and other vertices have degree $4k$ for some $k \geq 1$ (i.e., $c_{4,2}(G)=2$).
We will show that Alice can win the game on $G$ with the starting vertex $s$.

\begin{figure}[htb]
\centering
\unitlength 0.1in
\begin{picture}( 63.0000, 23.4500)(  5.4500,-34.5000)
%
\special{pn 8}%
\special{sh 1.000}%
\special{ar 800 3400 50 50  0.0000000 6.2831853}%
%
\special{pn 8}%
\special{sh 1.000}%
\special{ar 3600 3400 50 50  0.0000000 6.2831853}%
%
\special{pn 8}%
\special{sh 1.000}%
\special{ar 2200 1200 50 50  0.0000000 6.2831853}%
%
\special{pn 8}%
\special{pa 3600 3400}%
\special{pa 800 3400}%
\special{fp}%
%
\special{pn 8}%
\special{pa 800 3400}%
\special{pa 2200 1200}%
\special{fp}%
%
\special{pn 8}%
\special{pa 2200 1200}%
\special{pa 3600 3400}%
\special{fp}%
%
\special{pn 8}%
\special{pa 3600 3400}%
\special{pa 2200 3200}%
\special{fp}%
%
\special{pn 8}%
\special{pa 2200 3200}%
\special{pa 800 3400}%
\special{fp}%
%
\special{pn 8}%
\special{pa 800 3400}%
\special{pa 1600 2400}%
\special{fp}%
%
\special{pn 8}%
\special{pa 1600 2400}%
\special{pa 2200 1200}%
\special{fp}%
%
\special{pn 8}%
\special{pa 2200 1200}%
\special{pa 2800 2400}%
\special{fp}%
%
\special{pn 8}%
\special{pa 2800 2400}%
\special{pa 3600 3400}%
\special{fp}%
%
\special{pn 8}%
\special{pa 2800 2400}%
\special{pa 2200 3200}%
\special{fp}%
%
\special{pn 8}%
\special{pa 2200 3200}%
\special{pa 1600 2400}%
\special{fp}%
%
\special{pn 8}%
\special{pa 1600 2400}%
\special{pa 2800 2400}%
\special{fp}%
%
\special{pn 8}%
\special{pa 2200 2600}%
\special{pa 2000 2800}%
\special{fp}%
%
\special{pn 8}%
\special{pa 2000 2800}%
\special{pa 2400 2800}%
\special{fp}%
%
\special{pn 8}%
\special{pa 2400 2800}%
\special{pa 2200 2600}%
\special{fp}%
%
\special{pn 8}%
\special{pa 2200 2600}%
\special{pa 1600 2400}%
\special{fp}%
%
\special{pn 8}%
\special{pa 1600 2400}%
\special{pa 2000 2800}%
\special{fp}%
%
\special{pn 8}%
\special{pa 2000 2800}%
\special{pa 2200 3200}%
\special{fp}%
%
\special{pn 8}%
\special{pa 2200 3200}%
\special{pa 2400 2800}%
\special{fp}%
%
\special{pn 8}%
\special{pa 2400 2800}%
\special{pa 2800 2400}%
\special{fp}%
%
\special{pn 8}%
\special{pa 2800 2400}%
\special{pa 2200 2600}%
\special{fp}%
%
\special{pn 8}%
\special{pa 2200 2200}%
\special{pa 2000 2000}%
\special{fp}%
%
\special{pn 8}%
\special{pa 2000 2000}%
\special{pa 2400 2000}%
\special{fp}%
%
\special{pn 8}%
\special{pa 2400 2000}%
\special{pa 2200 2200}%
\special{fp}%
%
\special{pn 8}%
\special{pa 2200 2200}%
\special{pa 1600 2400}%
\special{fp}%
%
\special{pn 8}%
\special{pa 2800 2400}%
\special{pa 2200 2200}%
\special{fp}%
%
\special{pn 8}%
\special{pa 2000 2000}%
\special{pa 1600 2400}%
\special{fp}%
%
\special{pn 8}%
\special{pa 2000 2000}%
\special{pa 2200 1200}%
\special{fp}%
%
\special{pn 8}%
\special{pa 2200 1200}%
\special{pa 2400 2000}%
\special{fp}%
%
\special{pn 8}%
\special{pa 2400 2000}%
\special{pa 2800 2400}%
\special{fp}%
%
\special{pn 8}%
\special{sh 1.000}%
\special{ar 2200 3200 50 50  0.0000000 6.2831853}%
%
\special{pn 8}%
\special{sh 1.000}%
\special{ar 2800 2400 50 50  0.0000000 6.2831853}%
%
\special{pn 8}%
\special{sh 1.000}%
\special{ar 1600 2400 50 50  0.0000000 6.2831853}%
%
\special{pn 8}%
\special{sh 1.000}%
\special{ar 2000 2000 50 50  0.0000000 6.2831853}%
%
\special{pn 8}%
\special{sh 1.000}%
\special{ar 2400 2000 50 50  0.0000000 6.2831853}%
%
\special{pn 8}%
\special{sh 1.000}%
\special{ar 2200 2200 50 50  0.0000000 6.2831853}%
%
\special{pn 8}%
\special{sh 1.000}%
\special{ar 2200 2600 50 50  0.0000000 6.2831853}%
%
\special{pn 8}%
\special{sh 1.000}%
\special{ar 2000 2800 50 50  0.0000000 6.2831853}%
%
\special{pn 8}%
\special{sh 1.000}%
\special{ar 2400 2800 50 50  0.0000000 6.2831853}%
\put(6.8000,-33.0000){\makebox(0,0){$s$}}%
%
\special{pn 8}%
\special{sh 1.000}%
\special{ar 3996 3400 50 50  0.0000000 6.2831853}%
%
\special{pn 8}%
\special{sh 1.000}%
\special{ar 6796 3400 50 50  0.0000000 6.2831853}%
%
\special{pn 8}%
\special{sh 1.000}%
\special{ar 5396 1200 50 50  0.0000000 6.2831853}%
%
\special{pn 8}%
\special{sh 1.000}%
\special{ar 5996 2400 50 50  0.0000000 6.2831853}%
%
\special{pn 8}%
\special{sh 1.000}%
\special{ar 4796 2400 50 50  0.0000000 6.2831853}%
%
\special{pn 8}%
\special{sh 1.000}%
\special{ar 5196 2000 50 50  0.0000000 6.2831853}%
%
\special{pn 8}%
\special{sh 1.000}%
\special{ar 5596 2000 50 50  0.0000000 6.2831853}%
%
\special{pn 8}%
\special{sh 1.000}%
\special{ar 5396 2200 50 50  0.0000000 6.2831853}%
%
\special{pn 8}%
\special{sh 1.000}%
\special{ar 5196 2800 50 50  0.0000000 6.2831853}%
%
\special{pn 8}%
\special{sh 1.000}%
\special{ar 5596 2800 50 50  0.0000000 6.2831853}%
\put(38.7500,-33.0000){\makebox(0,0){$s$}}%
%
\special{pn 8}%
\special{pa 3990 3400}%
\special{pa 5350 3210}%
\special{fp}%
\special{sh 1}%
\special{pa 5350 3210}%
\special{pa 5282 3200}%
\special{pa 5298 3218}%
\special{pa 5288 3240}%
\special{pa 5350 3210}%
\special{fp}%
%
\special{pn 8}%
\special{pa 5200 2800}%
\special{pa 5400 3200}%
\special{fp}%
%
\special{pn 8}%
\special{pa 5200 2800}%
\special{pa 5400 2600}%
\special{fp}%
%
\special{pn 8}%
\special{pa 5400 2600}%
\special{pa 5600 2800}%
\special{fp}%
%
\special{pn 8}%
\special{pa 5600 2800}%
\special{pa 5400 3200}%
\special{fp}%
%
\special{pn 8}%
\special{pa 5400 3200}%
\special{pa 6000 2400}%
\special{fp}%
%
\special{pn 8}%
\special{pa 6000 2400}%
\special{pa 5400 2600}%
\special{fp}%
%
\special{pn 8}%
\special{pa 5400 3200}%
\special{pa 6730 3390}%
\special{dt 0.045}%
\special{sh 1}%
\special{pa 6730 3390}%
\special{pa 6668 3362}%
\special{pa 6678 3382}%
\special{pa 6662 3400}%
\special{pa 6730 3390}%
\special{fp}%
%
\special{pn 8}%
\special{sh 0}%
\special{ar 5400 3200 50 50  0.0000000 6.2831853}%
%
\special{pn 8}%
\special{sh 0}%
\special{ar 5400 2600 50 50  0.0000000 6.2831853}%
%
\special{pn 8}%
\special{pa 5346 2590}%
\special{pa 4866 2430}%
\special{dt 0.045}%
\special{sh 1}%
\special{pa 4866 2430}%
\special{pa 4922 2470}%
\special{pa 4916 2448}%
\special{pa 4936 2432}%
\special{pa 4866 2430}%
\special{fp}%
%
\special{pn 8}%
\special{pa 5336 3180}%
\special{pa 4836 2480}%
\special{dt 0.045}%
\special{sh 1}%
\special{pa 4836 2480}%
\special{pa 4858 2546}%
\special{pa 4866 2524}%
\special{pa 4890 2524}%
\special{pa 4836 2480}%
\special{fp}%
\put(22.9000,-33.0000){\makebox(0,0){$u_1$}}%
\put(15.8000,-25.7000){\makebox(0,0){$u_3$}}%
\put(27.9000,-25.9000){\makebox(0,0){$u_2$}}%
\put(22.0000,-24.8000){\makebox(0,0){$v_1$}}%
\put(20.0000,-26.8000){\makebox(0,0){$v_2$}}%
\put(24.0000,-26.8000){\makebox(0,0){$v_3$}}%
\put(23.5000,-11.9000){\makebox(0,0){$b$}}%
\put(36.2000,-35.3000){\makebox(0,0){$a$}}%
\put(19.4000,-18.7000){\makebox(0,0){$t_2$}}%
\put(22.0000,-23.2000){\makebox(0,0){$t_1$}}%
\put(22.8000,-18.9000){\makebox(0,0){$t_3$}}%
\end{picture}%
\caption{An Eulerian triangulation on the sphere with the starting vertex $s$ where Alice wins the game}
\label{basetri}
\end{figure}
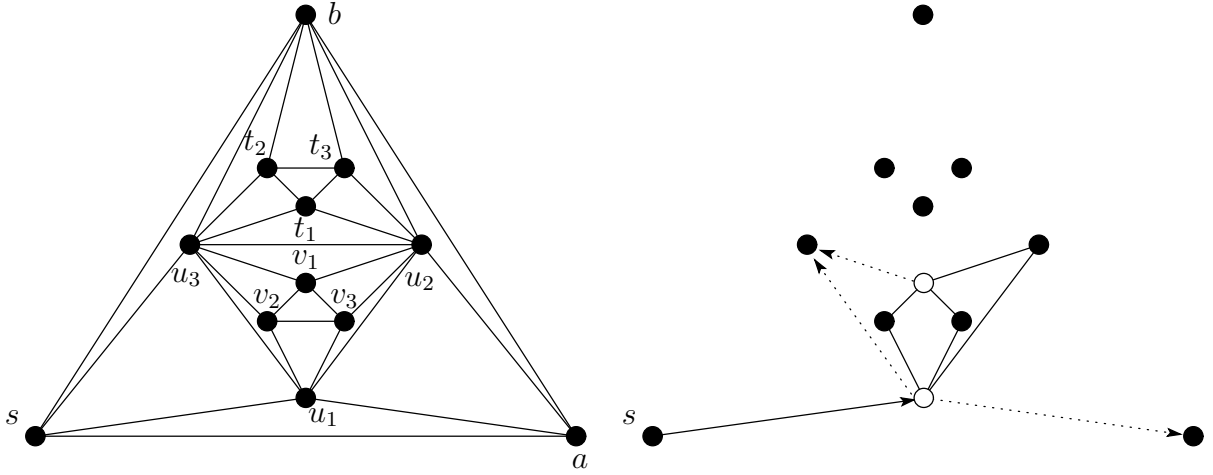

Alice first moves the token to $u_1$. 
After that, unless Bob moves the token to $a$ or $u_3$,
Alice moves it to a white vertex shown in the right of Figure~\ref{basetri}
corresponding to Bob's move.
As in the argument of an even kernel (see Remark~\ref{rem:even-kernel}),
Bob finally has to move the token to $a$ or $u_3$,
and hence, Alice wins the game.

\noindent
{\bf (The third step)}: 
Observe that four vertices $t_1,t_2,t_3$ and $b$ are not used for the above winning strategy of Alice in the second step. 
So, by repeatedly applying an octahedron addition to a face consisting three vertices of degree $4$ in the interior of $t_1t_2t_3$, 
we can increase the number of vertices of degree $6$ by $3r$ for $r\geq 1$ per one octahedron addition, 
keeping that Alice wins the game on the resulting graph. 
Thus the resulting graph is the desired graph for any $m \geq 2$ with $m \equiv 2 \pmod{3}$.

For any $m \geq 2$ with $m \equiv 1 \pmod{3}$, we construct the desired graph, as follows: 
First apply an octahedron addition to two faces $bu_3t_2$ and $bt_2t_3$ respectively. 
The resulting graph has four vertices $b,u_1,u_3$ and $t_3$ of degree $4k+2$ for some $k \geq 1$ and Alice still wins the game on the graph. 
After that, similar to the case when $m \equiv 2 \pmod{3}$, we can obtain the desired graphs 
by repeatedly applying an octahedron addition to a face consisting of three vertices of degree $4$ (in the interior of $bt_2t_3$).
\end{proof}

We conclude this section with the following result.

\begin{prop}\label{prop:dw}
Bob wins the game on every Eulerian double wheel graph. 
\end{prop}
\begin{proof}
Let $G$ be an Eulerian double wheel graph of $n+2$ vertices 
with $n\geq 4$ and $n \equiv 0 \pmod{2}$.
Let $C_n = v_0v_1 \dots v_{n-1}$ be the rim of $G$ 
and let $x,y$ be two vertices of degree $n$.
If $n \equiv 0 \pmod{4}$,
then Bob wins the game by Theorem~\ref{thm:deg4}.
Moreover,
if $x$ or $y$ is the starting vertex of $G$,
then Bob wins the game since $\{x,y\}$ is an even kernel.
Thus we may assume that
$n \equiv 2 \pmod{4}$ and $v_0$ is the starting vertex by symmetry.

If Alice first moves the token from $v_0$ to $v_1$ by symmetry,
then Bob moves it to $v_2$.
After that,
since Alice can move the token to neither $x$ nor $y$,
she has to move it from $v_{2i}$ to $v_{2i+1}$
and then Bob moves it to $v_{2i+2}$,
and hence, Bob can finally move the token back to $v_0$.

So Alice first moves the token from $v_0$ to $x$ by symmetry.
In this case, Bob first moves the token to $v_2$,
and then Alice moves it to $v_3$ and Bob moves it to $v_4$.
After that, Bob can win the game as follows:
If Alice moves the token to $x$,
then Bob moves it to $v_3$.
Since two edges $v_2v_3$ and $v_3v_4$ are already removed,
Alice must move the token to $y$ 
and then Bob can move it back to $v_0$.
Otherwise, similarly to the previous case, 
Bob can finally move the token back to $v_0$.
\end{proof}


\section{Octahedral Path}\label{sec:octa}

This section is devoted to proving Theorem~\ref{thm:octapath}. 
We use the same labeling of the vertices of $E_n$ as drawn in Figure~\ref{fig1}. 
We can see that $E_n$ can be also drawn as Figure~\ref{fig2}. 
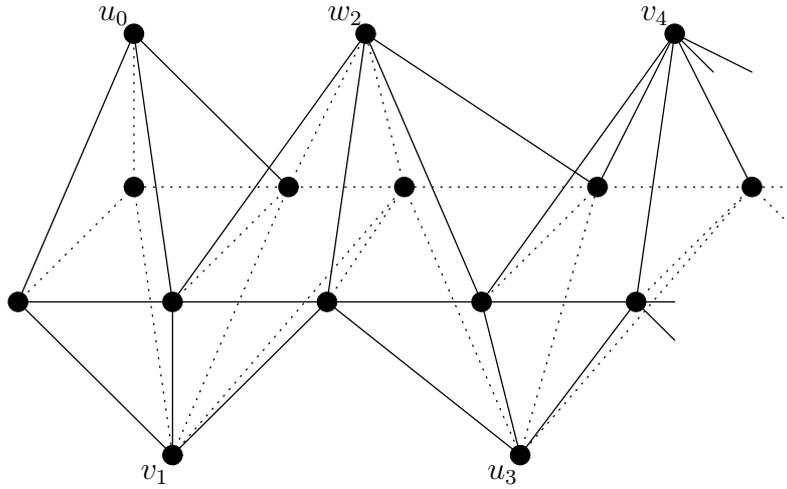
\begin{figure}[htb!]
\centering
\unitlength 0.1in
\begin{picture}( 40.5000, 24.3000)(  5.5000,-28.5000)
%
\special{pn 8}%
\special{sh 1.000}%
\special{ar 1400 2800 50 50  0.0000000 6.2831853}%
\put(10.9500,-5.0500){\makebox(0,0){$u_0$}}%
%
\special{pn 8}%
\special{sh 1.000}%
\special{ar 600 2000 50 50  0.0000000 6.2831853}%
%
\special{pn 8}%
\special{sh 1.000}%
\special{ar 1200 1400 50 50  0.0000000 6.2831853}%
%
\special{pn 8}%
\special{sh 1.000}%
\special{ar 1200 600 50 50  0.0000000 6.2831853}%
%
\special{pn 8}%
\special{sh 1.000}%
\special{ar 1400 2000 50 50  0.0000000 6.2831853}%
%
\special{pn 8}%
\special{sh 1.000}%
\special{ar 2000 1400 50 50  0.0000000 6.2831853}%
%
\special{pn 8}%
\special{sh 1.000}%
\special{ar 2200 2000 50 50  0.0000000 6.2831853}%
%
\special{pn 8}%
\special{sh 1.000}%
\special{ar 3200 2800 50 50  0.0000000 6.2831853}%
%
\special{pn 8}%
\special{sh 1.000}%
\special{ar 3000 2000 50 50  0.0000000 6.2831853}%
%
\special{pn 8}%
\special{sh 1.000}%
\special{ar 2600 1400 50 50  0.0000000 6.2831853}%
%
\special{pn 8}%
\special{sh 1.000}%
\special{ar 2400 600 50 50  0.0000000 6.2831853}%
%
\special{pn 8}%
\special{sh 1.000}%
\special{ar 3600 1400 50 50  0.0000000 6.2831853}%
%
\special{pn 8}%
\special{sh 1.000}%
\special{ar 4000 600 50 50  0.0000000 6.2831853}%
%
\special{pn 8}%
\special{sh 1.000}%
\special{ar 3800 2000 50 50  0.0000000 6.2831853}%
%
\special{pn 8}%
\special{sh 1.000}%
\special{ar 4400 1400 50 50  0.0000000 6.2831853}%
%
\special{pn 8}%
\special{pa 600 2000}%
\special{pa 4000 2000}%
\special{fp}%
%
\special{pn 8}%
\special{pa 1200 600}%
\special{pa 600 2000}%
\special{fp}%
%
\special{pn 8}%
\special{pa 1400 2000}%
\special{pa 1200 600}%
\special{fp}%
%
\special{pn 8}%
\special{pa 2000 1400}%
\special{pa 1200 600}%
\special{fp}%
%
\special{pn 8}%
\special{pa 2400 600}%
\special{pa 1400 2000}%
\special{fp}%
%
\special{pn 8}%
\special{pa 2200 2000}%
\special{pa 2400 600}%
\special{fp}%
%
\special{pn 8}%
\special{pa 2400 600}%
\special{pa 3600 1400}%
\special{fp}%
%
\special{pn 8}%
\special{pa 4000 600}%
\special{pa 3600 1400}%
\special{fp}%
%
\special{pn 8}%
\special{pa 3000 2000}%
\special{pa 2400 600}%
\special{fp}%
%
\special{pn 8}%
\special{pa 4000 600}%
\special{pa 3000 2000}%
\special{fp}%
%
\special{pn 8}%
\special{pa 3800 2000}%
\special{pa 4000 600}%
\special{fp}%
%
\special{pn 8}%
\special{pa 4000 600}%
\special{pa 4400 1400}%
\special{fp}%
%
\special{pn 8}%
\special{pa 4600 1400}%
\special{pa 1200 1400}%
\special{dt 0.045}%
%
\special{pn 8}%
\special{pa 1200 1400}%
\special{pa 600 2000}%
\special{dt 0.045}%
%
\special{pn 8}%
\special{pa 1200 1400}%
\special{pa 1200 600}%
\special{dt 0.045}%
%
\special{pn 8}%
\special{pa 2000 1400}%
\special{pa 2400 600}%
\special{dt 0.045}%
%
\special{pn 8}%
\special{pa 1400 2000}%
\special{pa 2000 1400}%
\special{dt 0.045}%
%
\special{pn 8}%
\special{pa 3800 2000}%
\special{pa 4400 1400}%
\special{dt 0.045}%
%
\special{pn 8}%
\special{pa 3600 1400}%
\special{pa 3000 2000}%
\special{dt 0.045}%
\special{pa 3000 2000}%
\special{pa 3000 2000}%
\special{dt 0.045}%
%
\special{pn 8}%
\special{pa 1400 2800}%
\special{pa 600 2000}%
\special{fp}%
%
\special{pn 8}%
\special{pa 1400 2000}%
\special{pa 1400 2800}%
\special{fp}%
%
\special{pn 8}%
\special{pa 1400 2800}%
\special{pa 2200 2000}%
\special{fp}%
%
\special{pn 8}%
\special{pa 2200 2000}%
\special{pa 3200 2800}%
\special{fp}%
%
\special{pn 8}%
\special{pa 3200 2800}%
\special{pa 3000 2000}%
\special{fp}%
%
\special{pn 8}%
\special{pa 3200 2800}%
\special{pa 3800 2000}%
\special{fp}%
%
\special{pn 8}%
\special{pa 3800 2000}%
\special{pa 4000 2200}%
\special{fp}%
%
\special{pn 8}%
\special{pa 3200 2800}%
\special{pa 3600 1400}%
\special{dt 0.045}%
%
\special{pn 8}%
\special{pa 4400 1400}%
\special{pa 3200 2800}%
\special{dt 0.045}%
%
\special{pn 8}%
\special{pa 1400 2800}%
\special{pa 2000 1400}%
\special{dt 0.045}%
%
\special{pn 8}%
\special{pa 1200 1400}%
\special{pa 1400 2800}%
\special{dt 0.045}%
%
\special{pn 8}%
\special{pa 2600 1400}%
\special{pa 2400 600}%
\special{dt 0.045}%
%
\special{pn 8}%
\special{pa 2600 1400}%
\special{pa 2200 2000}%
\special{dt 0.045}%
%
\special{pn 8}%
\special{pa 2600 1400}%
\special{pa 3200 2800}%
\special{dt 0.045}%
%
\special{pn 8}%
\special{pa 1400 2800}%
\special{pa 2600 1400}%
\special{dt 0.045}%
\put(13.1000,-29.0000){\makebox(0,0){$v_1$}}%
\put(31.1000,-29.0000){\makebox(0,0){$u_3$}}%
\put(22.9500,-5.0500){\makebox(0,0){$w_2$}}%
\put(38.9500,-5.0500){\makebox(0,0){$v_4$}}%
%
\special{pn 8}%
\special{pa 4000 600}%
\special{pa 4200 800}%
\special{fp}%
%
\special{pn 8}%
\special{pa 4000 600}%
\special{pa 4400 800}%
\special{fp}%
%
\special{pn 8}%
\special{pa 4400 1400}%
\special{pa 4600 1600}%
\special{dt 0.045}%
\end{picture}%
\caption{Another drawing of the octahedral path}\label{fig2}
\end{figure}


\begin{lem}\label{lem1}
Alice (resp., Bob) can win the game on $E(n,p)$ if and only if Alice (resp., Bob) can win the game on $E(n,n-p)$. 
\end{lem}
\begin{proof}
It is trivial by the symmetry of the vertices $u_p,v_p,w_p$ and $u_{n-p},v_{n-p},w_{n-p}$ on $E_n$, respectively. 
\end{proof}

For the proof of the lemmas below, we describe the neighbors of each vertex: 
\begin{equation}\label{neighbor}
\begin{split}
&N(u_i)=\begin{cases}
\{v_i,w_i,u_{i+1},w_{i+1}\} &\text{ if }i=0, \\
\{v_i,w_i,u_{i-1},v_{i-1}\} &\text{ if }i=n, \\
\{v_i,w_i,u_{i+1},w_{i+1},u_{i-1},v_{i-1}\} &\text{ otherwise}, 
\end{cases}\\
&N(v_i)=\begin{cases}
\{u_i,w_i,v_{i+1},u_{i+1}\} &\text{ if }i=0, \\
\{u_i,w_i,v_{i-1},w_{i-1}\} &\text{ if }i=n, \\
\{u_i,w_i,v_{i+1},u_{i+1},v_{i-1},w_{i-1}\} &\text{ otherwise}, 
\end{cases}\\
&N(w_i)=\begin{cases}
\{u_i,v_i,w_{i+1},v_{i+1}\} &\text{ if }i=0, \\
\{u_i,v_i,w_{i-1},u_{i-1}\} &\text{ if }i=n, \\
\{u_i,v_i,w_{i+1},v_{i+1},w_{i-1},u_{i-1}\} &\text{ otherwise}. 
\end{cases}
\end{split}
\end{equation}

\begin{lem}\label{lem2}
Bob can win the game on $E(3m+1,3k)$ and $E(3m+1,3k+1)$ 
for any $m,k \in \ZZ_{\geq 0}$. 
\end{lem}
\begin{proof}
Let $$S=\{u_{3i},v_{3i+1} : i=0,1,\ldots,m\}.$$
For proving the statement, it is enough to show that $S$ is an even kernel. 
By \eqref{neighbor}, we can easily see the following: 
\begin{align*}
&N(v_{3i}) \cap S =N(w_{3i}) \cap S =\{u_{3i},v_{3i+1}\} \text{ for }i=0,1,\ldots,m; \\
&N(u_{3i+1}) \cap S =N(w_{3i+1}) \cap S=\{u_{3i},v_{3i+1}\} \text{ for }i=0,1,\ldots,m; \\
&N(u_{3i+2}) \cap S=N(v_{3i+2}) \cap S=\{v_{3i+1},u_{3i+3}\} \text{ for }i=0,1,\ldots,m-1; \\
&N(w_{3i+2}) \cap S=\emptyset. 
\end{align*}
\end{proof}

\begin{lem}\label{lem3}
Bob can win the game on $E(3m+2,3k+1)$ for any 
$m,k \in \ZZ_{\geq 0}$. 
\end{lem}
\begin{proof}
Let us consider $S=S_1 \cup S_2$, where 
$$S_1=\{u_{3i},v_{3i+1} : i=0,1,\ldots,k\} \text{ and } S_2=\{v_{3k+1+3j},w_{3k+2+3j} : j=0,1,\ldots,m-k\}.$$ 
Consider the bipartite graph $H$ as depicted in Figure~\ref{fig3}: 
\begin{figure}[htb!]
\centering
\unitlength 0.1in
\begin{picture}( 64.5000, 27.5000)( 13.5000,-41.5500)
%
\special{pn 8}%
\special{sh 1.000}%
\special{ar 4600 3800 50 50  0.0000000 6.2831853}%
\put(43.0000,-14.9000){\makebox(0,0){$v_{3k+1}$}}%
%
\special{pn 8}%
\special{sh 1.000}%
\special{ar 4400 1600 50 50  0.0000000 6.2831853}%
%
\special{pn 8}%
\special{sh 1.000}%
\special{ar 6400 3800 50 50  0.0000000 6.2831853}%
%
\special{pn 8}%
\special{sh 0}%
\special{ar 5800 1600 50 50  0.0000000 6.2831853}%
%
\special{pn 8}%
\special{sh 1.000}%
\special{ar 7200 1600 50 50  0.0000000 6.2831853}%
%
\special{pn 8}%
\special{pa 4600 3000}%
\special{pa 4400 1600}%
\special{fp}%
%
\special{pn 8}%
\special{pa 5200 2400}%
\special{pa 4400 1600}%
\special{fp}%
%
\special{pn 8}%
\special{pa 7200 1600}%
\special{pa 6800 2400}%
\special{fp}%
%
\special{pn 8}%
\special{pa 7200 1600}%
\special{pa 6200 3000}%
\special{fp}%
%
\special{pn 8}%
\special{pa 7000 3000}%
\special{pa 7200 1600}%
\special{fp}%
%
\special{pn 8}%
\special{pa 7200 1600}%
\special{pa 7600 2400}%
\special{fp}%
%
\special{pn 8}%
\special{pa 4400 2400}%
\special{pa 4400 1600}%
\special{fp}%
%
\special{pn 8}%
\special{pa 4600 3000}%
\special{pa 4600 3800}%
\special{fp}%
%
\special{pn 8}%
\special{pa 4600 3800}%
\special{pa 5400 3000}%
\special{fp}%
%
\special{pn 8}%
\special{pa 5400 3000}%
\special{pa 6400 3800}%
\special{fp}%
%
\special{pn 8}%
\special{pa 6400 3800}%
\special{pa 6200 3000}%
\special{fp}%
%
\special{pn 8}%
\special{pa 6400 3800}%
\special{pa 7000 3000}%
\special{fp}%
%
\special{pn 8}%
\special{pa 6400 3800}%
\special{pa 6800 2400}%
\special{fp}%
%
\special{pn 8}%
\special{pa 7600 2400}%
\special{pa 6400 3800}%
\special{fp}%
%
\special{pn 8}%
\special{pa 4600 3800}%
\special{pa 5200 2400}%
\special{fp}%
%
\special{pn 8}%
\special{pa 4400 2400}%
\special{pa 4600 3800}%
\special{fp}%
%
\special{pn 8}%
\special{pa 5800 2400}%
\special{pa 6400 3800}%
\special{fp}%
%
\special{pn 8}%
\special{pa 4600 3800}%
\special{pa 5800 2400}%
\special{fp}%
\put(45.1000,-39.0000){\makebox(0,0){$w_{3k+2}$}}%
\put(43.0000,-23.1000){\makebox(0,0){$y$}}%
\put(53.0000,-23.1000){\makebox(0,0){$z$}}%
%
\special{pn 8}%
\special{pa 7200 1600}%
\special{pa 7400 1800}%
\special{fp}%
%
\special{pn 8}%
\special{pa 7200 1600}%
\special{pa 7600 1800}%
\special{fp}%
%
\special{pn 8}%
\special{pa 3600 2400}%
\special{pa 4400 1600}%
\special{fp}%
%
\special{pn 8}%
\special{pa 4400 1600}%
\special{pa 3200 3000}%
\special{fp}%
%
\special{pn 8}%
\special{pa 3200 3000}%
\special{pa 3000 3800}%
\special{fp}%
%
\special{pn 8}%
\special{pa 3000 3800}%
\special{pa 3600 2400}%
\special{fp}%
%
\special{pn 8}%
\special{pa 3000 3800}%
\special{pa 2600 3000}%
\special{fp}%
%
\special{pn 8}%
\special{pa 3000 3800}%
\special{pa 2800 2400}%
\special{fp}%
%
\special{pn 8}%
\special{pa 2800 2400}%
\special{pa 2000 3800}%
\special{fp}%
%
\special{pn 8}%
\special{pa 2000 3800}%
\special{pa 2600 3000}%
\special{fp}%
%
\special{pn 8}%
\special{pa 2000 1600}%
\special{pa 1600 1800}%
\special{fp}%
%
\special{pn 8}%
\special{pa 1800 1800}%
\special{pa 2000 1600}%
\special{fp}%
%
\special{pn 8}%
\special{sh 1.000}%
\special{ar 2000 1600 50 50  0.0000000 6.2831853}%
%
\special{pn 8}%
\special{pa 1600 2400}%
\special{pa 2000 1600}%
\special{fp}%
%
\special{pn 8}%
\special{pa 2000 3800}%
\special{pa 1600 2400}%
\special{fp}%
%
\special{pn 8}%
\special{pa 1400 3000}%
\special{pa 2000 3800}%
\special{fp}%
%
\special{pn 8}%
\special{pa 1400 3000}%
\special{pa 2000 1600}%
\special{fp}%
%
\special{pn 8}%
\special{sh 0}%
\special{ar 5200 2400 50 50  0.0000000 6.2831853}%
%
\special{pn 8}%
\special{sh 0}%
\special{ar 4600 3000 50 50  0.0000000 6.2831853}%
%
\special{pn 8}%
\special{sh 0}%
\special{ar 3200 3000 50 50  0.0000000 6.2831853}%
%
\special{pn 8}%
\special{sh 0}%
\special{ar 2600 3000 50 50  0.0000000 6.2831853}%
%
\special{pn 8}%
\special{sh 0}%
\special{ar 1400 3000 50 50  0.0000000 6.2831853}%
%
\special{pn 8}%
\special{sh 0}%
\special{ar 1600 2400 50 50  0.0000000 6.2831853}%
%
\special{pn 8}%
\special{sh 0}%
\special{ar 2800 2400 50 50  0.0000000 6.2831853}%
%
\special{pn 8}%
\special{sh 0}%
\special{ar 3600 2400 50 50  0.0000000 6.2831853}%
%
\special{pn 8}%
\special{sh 0}%
\special{ar 5800 2400 50 50  0.0000000 6.2831853}%
%
\special{pn 8}%
\special{sh 0}%
\special{ar 6800 2400 50 50  0.0000000 6.2831853}%
%
\special{pn 8}%
\special{sh 0}%
\special{ar 7600 2400 50 50  0.0000000 6.2831853}%
%
\special{pn 8}%
\special{sh 0}%
\special{ar 7000 3000 50 50  0.0000000 6.2831853}%
%
\special{pn 8}%
\special{sh 0}%
\special{ar 6200 3000 50 50  0.0000000 6.2831853}%
%
\special{pn 8}%
\special{sh 0}%
\special{ar 5400 3000 50 50  0.0000000 6.2831853}%
\put(35.0000,-23.1000){\makebox(0,0){$x$}}%
%
\special{pn 8}%
\special{sh 1.000}%
\special{ar 3000 3800 50 50  0.0000000 6.2831853}%
%
\special{pn 8}%
\special{sh 1.000}%
\special{ar 2000 3800 50 50  0.0000000 6.2831853}%
%
\special{pn 8}%
\special{pa 2000 3000}%
\special{pa 2000 3800}%
\special{fp}%
%
\special{pn 8}%
\special{pa 2000 3000}%
\special{pa 2000 1600}%
\special{fp}%
%
\special{pn 8}%
\special{pa 2000 1600}%
\special{pa 2400 2400}%
\special{fp}%
%
\special{pn 8}%
\special{pa 2400 2400}%
\special{pa 2000 3800}%
\special{fp}%
%
\special{pn 8}%
\special{sh 0}%
\special{ar 2400 2400 50 50  0.0000000 6.2831853}%
%
\special{pn 8}%
\special{sh 0}%
\special{ar 2000 3000 50 50  0.0000000 6.2831853}%
\put(29.1000,-39.0000){\makebox(0,0){$u_{3k}$}}%
%
\special{pn 8}%
\special{sh 0}%
\special{ar 3000 1600 50 50  0.0000000 6.2831853}%
%
\special{pn 8}%
\special{pa 3970 3300}%
\special{pa 3600 4100}%
\special{fp}%
%
\special{pn 8}%
\special{pa 4200 4100}%
\special{pa 7400 4100}%
\special{fp}%
%
\special{pn 8}%
\special{pa 4200 4100}%
\special{pa 4030 3300}%
\special{fp}%
\put(26.0000,-42.4000){\makebox(0,0){$H_1$}}%
\put(56.0000,-42.4000){\makebox(0,0){$H_2$}}%
%
\special{pn 8}%
\special{pa 3600 4100}%
\special{pa 1800 4100}%
\special{fp}%
%
\special{pn 8}%
\special{pa 1800 4100}%
\special{pa 1600 4100}%
\special{dt 0.045}%
%
\special{pn 8}%
\special{pa 7400 4100}%
\special{pa 7800 4100}%
\special{dt 0.045}%
%
\special{pn 8}%
\special{pa 3000 3800}%
\special{pa 4400 2400}%
\special{fp}%
%
\special{pn 8}%
\special{pa 4000 3000}%
\special{pa 4600 3800}%
\special{fp}%
%
\special{pn 8}%
\special{pa 4000 3000}%
\special{pa 4400 1600}%
\special{fp}%
%
\special{pn 8}%
\special{pa 4000 3000}%
\special{pa 3000 3800}%
\special{fp}%
%
\special{pn 8}%
\special{sh 0}%
\special{ar 4400 2400 50 50  0.0000000 6.2831853}%
%
\special{pn 8}%
\special{sh 0}%
\special{ar 4000 3000 50 50  0.0000000 6.2831853}%
\end{picture}%
\caption{``Almost'' even kernel graph $H$}\label{fig3}
\end{figure}
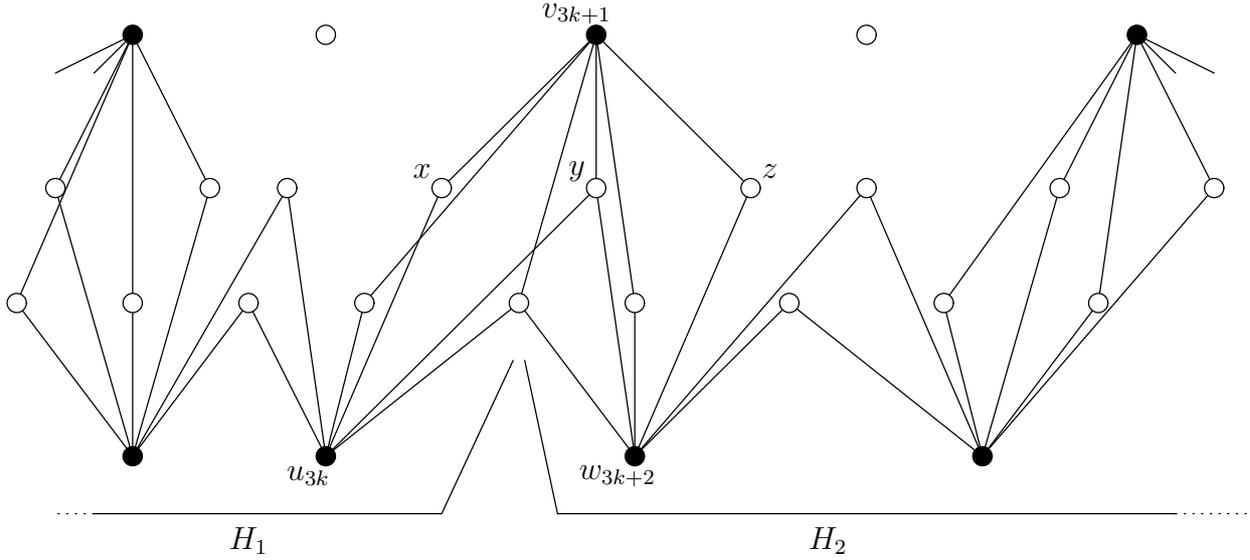

Note that there are two even kernel graphs $H_1$ and $H_2$
with respect to $S_1$ and $S_2$, respectively, 
although $S$ is not an even kernel since $u_{3k+1}$ and $w_{3k+1}$ are adjacent to $u_{3k},v_{3k+1}$ and $w_{3k+2}$. 

If Alice moves the token to the vertex $x$ or $y$ (resp., $z$ or $y$), then Bob may move the token to $u_{3k}$ (resp., $w_{3k+2}$). 
After that, by following the even kernel graph $H_1$ (resp., $H_2$), we see that Bob can move the token to the vertex $u_{3k}$ (resp., $w_{3k+2}$). 
Hence, Bob can finally win the game. 
\end{proof}

\begin{lem}\label{lem4}
Alice can win the game on $E(n,3k+2)$ for any $n,k \in \ZZ_{\geq 0}$. 
\end{lem}
\begin{proof}
Let us consider the same subset
as in Lemma~\ref{lem2}, i.e., $$S=\{u_{3i},v_{3i+1} : i=0,1,\ldots,k\},$$ 
and let $s=v_{3k+2}$ be the starting vertex (see Figure~\ref{fig4}). 
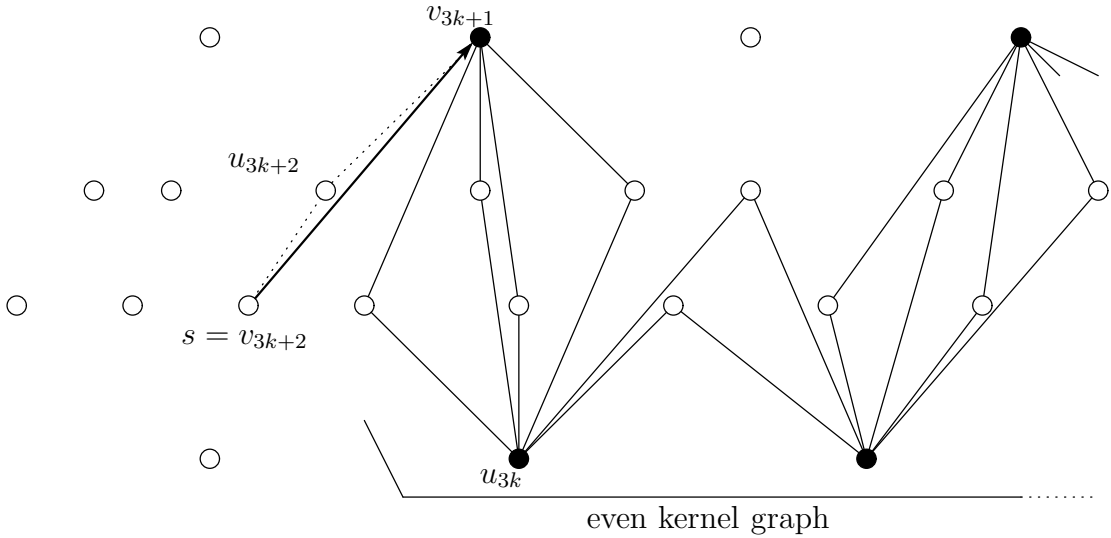
\begin{figure}[htb!]
\centering
\unitlength 0.1in
\begin{picture}( 57.0000, 26.3000)( 19.5000,-40.3500)
%
\special{pn 8}%
\special{sh 1.000}%
\special{ar 4600 3800 50 50  0.0000000 6.2831853}%
\put(43.0000,-14.9000){\makebox(0,0){$v_{3k+1}$}}%
%
\special{pn 8}%
\special{sh 1.000}%
\special{ar 4400 1600 50 50  0.0000000 6.2831853}%
%
\special{pn 8}%
\special{sh 1.000}%
\special{ar 6400 3800 50 50  0.0000000 6.2831853}%
%
\special{pn 8}%
\special{sh 0}%
\special{ar 5800 1600 50 50  0.0000000 6.2831853}%
%
\special{pn 8}%
\special{sh 1.000}%
\special{ar 7200 1600 50 50  0.0000000 6.2831853}%
%
\special{pn 8}%
\special{pa 4600 3000}%
\special{pa 4400 1600}%
\special{fp}%
%
\special{pn 8}%
\special{pa 5200 2400}%
\special{pa 4400 1600}%
\special{fp}%
%
\special{pn 8}%
\special{pa 7200 1600}%
\special{pa 6800 2400}%
\special{fp}%
%
\special{pn 8}%
\special{pa 7200 1600}%
\special{pa 6200 3000}%
\special{fp}%
%
\special{pn 8}%
\special{pa 7000 3000}%
\special{pa 7200 1600}%
\special{fp}%
%
\special{pn 8}%
\special{pa 7200 1600}%
\special{pa 7600 2400}%
\special{fp}%
%
\special{pn 8}%
\special{pa 4400 2400}%
\special{pa 4400 1600}%
\special{fp}%
%
\special{pn 8}%
\special{pa 4600 3000}%
\special{pa 4600 3800}%
\special{fp}%
%
\special{pn 8}%
\special{pa 4600 3800}%
\special{pa 5400 3000}%
\special{fp}%
%
\special{pn 8}%
\special{pa 5400 3000}%
\special{pa 6400 3800}%
\special{fp}%
%
\special{pn 8}%
\special{pa 6400 3800}%
\special{pa 6200 3000}%
\special{fp}%
%
\special{pn 8}%
\special{pa 6400 3800}%
\special{pa 7000 3000}%
\special{fp}%
%
\special{pn 8}%
\special{pa 6400 3800}%
\special{pa 6800 2400}%
\special{fp}%
%
\special{pn 8}%
\special{pa 7600 2400}%
\special{pa 6400 3800}%
\special{fp}%
%
\special{pn 8}%
\special{pa 4600 3800}%
\special{pa 5200 2400}%
\special{fp}%
%
\special{pn 8}%
\special{pa 4400 2400}%
\special{pa 4600 3800}%
\special{fp}%
%
\special{pn 8}%
\special{pa 5800 2400}%
\special{pa 6400 3800}%
\special{fp}%
%
\special{pn 8}%
\special{pa 4600 3800}%
\special{pa 5800 2400}%
\special{fp}%
\put(45.1000,-39.0000){\makebox(0,0){$u_{3k}$}}%
%
\special{pn 8}%
\special{pa 7200 1600}%
\special{pa 7400 1800}%
\special{fp}%
%
\special{pn 8}%
\special{pa 7200 1600}%
\special{pa 7600 1800}%
\special{fp}%
%
\special{pn 8}%
\special{sh 0}%
\special{ar 5200 2400 50 50  0.0000000 6.2831853}%
%
\special{pn 8}%
\special{sh 0}%
\special{ar 4400 2400 50 50  0.0000000 6.2831853}%
%
\special{pn 8}%
\special{sh 0}%
\special{ar 4600 3000 50 50  0.0000000 6.2831853}%
%
\special{pn 8}%
\special{sh 0}%
\special{ar 2600 3000 50 50  0.0000000 6.2831853}%
%
\special{pn 8}%
\special{sh 0}%
\special{ar 2800 2400 50 50  0.0000000 6.2831853}%
%
\special{pn 8}%
\special{sh 0}%
\special{ar 5800 2400 50 50  0.0000000 6.2831853}%
%
\special{pn 8}%
\special{sh 0}%
\special{ar 6800 2400 50 50  0.0000000 6.2831853}%
%
\special{pn 8}%
\special{sh 0}%
\special{ar 7600 2400 50 50  0.0000000 6.2831853}%
%
\special{pn 8}%
\special{sh 0}%
\special{ar 7000 3000 50 50  0.0000000 6.2831853}%
%
\special{pn 8}%
\special{sh 0}%
\special{ar 6200 3000 50 50  0.0000000 6.2831853}%
%
\special{pn 8}%
\special{sh 0}%
\special{ar 5400 3000 50 50  0.0000000 6.2831853}%
%
\special{pn 8}%
\special{sh 0}%
\special{ar 3000 3800 50 50  0.0000000 6.2831853}%
%
\special{pn 8}%
\special{sh 0}%
\special{ar 2400 2400 50 50  0.0000000 6.2831853}%
%
\special{pn 8}%
\special{sh 0}%
\special{ar 2000 3000 50 50  0.0000000 6.2831853}%
%
\special{pn 8}%
\special{sh 0}%
\special{ar 3000 1600 50 50  0.0000000 6.2831853}%
\put(31.8000,-31.8000){\makebox(0,0){$s = v_{3k+2}$}}%
%
\special{pn 13}%
\special{pa 3220 2980}%
\special{pa 4350 1640}%
\special{fp}%
\special{sh 1}%
\special{pa 4350 1640}%
\special{pa 4292 1678}%
\special{pa 4316 1682}%
\special{pa 4322 1704}%
\special{pa 4350 1640}%
\special{fp}%
\put(32.8000,-22.6000){\makebox(0,0){$u_{3k+2}$}}%
%
\special{pn 8}%
\special{pa 3600 2390}%
\special{pa 4390 1610}%
\special{dt 0.045}%
%
\special{pn 8}%
\special{pa 4600 3800}%
\special{pa 3800 3000}%
\special{fp}%
%
\special{pn 8}%
\special{pa 3800 3000}%
\special{pa 4400 1600}%
\special{fp}%
%
\special{pn 8}%
\special{sh 0}%
\special{ar 3800 3000 50 50  0.0000000 6.2831853}%
%
\special{pn 8}%
\special{pa 3800 3600}%
\special{pa 4000 4000}%
\special{fp}%
%
\special{pn 8}%
\special{pa 4000 4000}%
\special{pa 7200 4000}%
\special{fp}%
%
\special{pn 8}%
\special{pa 7200 4000}%
\special{pa 7600 4000}%
\special{dt 0.045}%
\put(55.8000,-41.2000){\makebox(0,0){even kernel graph}}%
%
\special{pn 8}%
\special{pa 3600 2400}%
\special{pa 3200 3000}%
\special{dt 0.045}%
%
\special{pn 8}%
\special{sh 0}%
\special{ar 3600 2400 50 50  0.0000000 6.2831853}%
%
\special{pn 8}%
\special{sh 0}%
\special{ar 3200 3000 50 50  0.0000000 6.2831853}%
\end{picture}%
\caption{The even kernel graph with respect to $S$ in Lemma~\ref{lem4}}\label{fig4}
\end{figure}

First, Alice moves the token to the vertex $v_{3k+1} \in S$. 
Then by following the even kernel graph with respect to $S$, 
we see that Alice can move the token back to $v_{3k+1}$ again. 
Finally, Bob must move the token to $u_{3k+2}$ from $v_{3k+1}$. 
Therefore, Alice can win the game. 
\end{proof}

\begin{lem}\label{lem5}
Alice can win the game on $E(3m,3k)$ for any $m,k \in \ZZ_{\geq 0}$. 
\end{lem}
\begin{proof}
Let us consider the same subsets as in Lemma~\ref{lem3}, i.e., 
\begin{align*}
S_1&=\{u_{3i},v_{3i+1} : i=0,1,\ldots,k-1\} \text{ and }\\ 
S_2&=\{w_{3k+2+3j}, u_{3k+2+3j+1} : j=0,1,\ldots,m-k-1\}, 
\end{align*}
and let $s=v_{3k}$ be the starting vertex (see Figure~\ref{fig5}). 
\begin{figure}[htb!]
\centering
\unitlength 0.1in
\begin{picture}( 66.0000, 24.5000)( 16.0000,-36.5000)
\put(46.9500,-12.8500){\makebox(0,0){$u_{3k}$}}%
%
\special{pn 8}%
\special{sh 1.000}%
\special{ar 6800 3600 50 50  0.0000000 6.2831853}%
%
\special{pn 8}%
\special{sh 0}%
\special{ar 7600 1400 50 50  0.0000000 6.2831853}%
%
\special{pn 8}%
\special{pa 5800 2800}%
\special{pa 6800 3600}%
\special{fp}%
%
\special{pn 8}%
\special{pa 6800 3600}%
\special{pa 6600 2800}%
\special{fp}%
%
\special{pn 8}%
\special{pa 6800 3600}%
\special{pa 7400 2800}%
\special{fp}%
%
\special{pn 8}%
\special{pa 7400 2800}%
\special{pa 7600 3000}%
\special{fp}%
%
\special{pn 8}%
\special{pa 6800 3600}%
\special{pa 7200 2200}%
\special{fp}%
%
\special{pn 8}%
\special{pa 8000 2200}%
\special{pa 6800 3600}%
\special{fp}%
%
\special{pn 8}%
\special{pa 6200 2200}%
\special{pa 6800 3600}%
\special{fp}%
\put(49.1000,-37.1000){\makebox(0,0){$v_{3k+1}$}}%
%
\special{pn 8}%
\special{pa 8000 2200}%
\special{pa 8200 2400}%
\special{fp}%
%
\special{pn 8}%
\special{sh 1.000}%
\special{ar 2400 3600 50 50  0.0000000 6.2831853}%
%
\special{pn 8}%
\special{pa 2400 2800}%
\special{pa 2400 3600}%
\special{fp}%
%
\special{pn 8}%
\special{pa 5000 2800}%
\special{pa 4200 2800}%
\special{fp}%
%
\special{pn 8}%
\special{pa 4200 2800}%
\special{pa 3600 2800}%
\special{fp}%
%
\special{pn 8}%
\special{pa 3600 2800}%
\special{pa 3400 1400}%
\special{fp}%
%
\special{pn 8}%
\special{pa 6200 1400}%
\special{pa 5600 2200}%
\special{dt 0.045}%
%
\special{pn 8}%
\special{pa 4000 2200}%
\special{pa 3400 1400}%
\special{dt 0.045}%
%
\special{pn 8}%
\special{pa 4800 1400}%
\special{pa 4600 2200}%
\special{dt 0.045}%
%
\special{pn 8}%
\special{pa 4600 2200}%
\special{pa 3400 3600}%
\special{dt 0.045}%
%
\special{pn 8}%
\special{pa 3400 3600}%
\special{pa 4200 2800}%
\special{dt 0.045}%
%
\special{pn 8}%
\special{pa 4200 2800}%
\special{pa 5000 3600}%
\special{dt 0.045}%
%
\special{pn 8}%
\special{pa 5000 3600}%
\special{pa 4600 2200}%
\special{dt 0.045}%
%
\special{pn 8}%
\special{pa 4600 2200}%
\special{pa 5600 2200}%
\special{dt 0.045}%
%
\special{pn 8}%
\special{pa 4600 2200}%
\special{pa 4000 2200}%
\special{dt 0.045}%
%
\special{pn 13}%
\special{pa 4600 2200}%
\special{pa 4250 2750}%
\special{fp}%
\special{sh 1}%
\special{pa 4250 2750}%
\special{pa 4304 2704}%
\special{pa 4280 2706}%
\special{pa 4270 2684}%
\special{pa 4250 2750}%
\special{fp}%
%
\special{pn 8}%
\special{sh 1.000}%
\special{ar 6200 1400 50 50  0.0000000 6.2831853}%
%
\special{pn 8}%
\special{sh 1.000}%
\special{ar 3400 1400 50 50  0.0000000 6.2831853}%
%
\special{pn 8}%
\special{pa 5000 2800}%
\special{pa 6200 1400}%
\special{fp}%
%
\special{pn 8}%
\special{pa 6200 1400}%
\special{pa 6200 2200}%
\special{fp}%
%
\special{pn 8}%
\special{pa 5800 2800}%
\special{pa 6200 1400}%
\special{fp}%
%
\special{pn 8}%
\special{pa 6200 1400}%
\special{pa 7200 2200}%
\special{fp}%
%
\special{pn 8}%
\special{pa 6200 1400}%
\special{pa 6600 2800}%
\special{fp}%
\put(69.1000,-37.3000){\makebox(0,0){$u_{3k+3}$}}%
%
\special{pn 8}%
\special{pa 3400 2200}%
\special{pa 3400 1400}%
\special{fp}%
%
\special{pn 8}%
\special{pa 3400 2200}%
\special{pa 2400 3600}%
\special{fp}%
%
\special{pn 8}%
\special{pa 3200 2800}%
\special{pa 3400 1400}%
\special{fp}%
%
\special{pn 8}%
\special{pa 3200 2800}%
\special{pa 2400 3600}%
\special{fp}%
%
\special{pn 8}%
\special{pa 2400 3600}%
\special{pa 1800 2800}%
\special{fp}%
%
\special{pn 8}%
\special{pa 1800 2800}%
\special{pa 1600 3000}%
\special{fp}%
%
\special{pn 8}%
\special{pa 2000 2200}%
\special{pa 1800 2400}%
\special{fp}%
%
\special{pn 8}%
\special{pa 2000 2200}%
\special{pa 2400 3600}%
\special{fp}%
%
\special{pn 8}%
\special{pa 2400 2800}%
\special{pa 3400 1400}%
\special{fp}%
%
\special{pn 8}%
\special{pa 3400 1400}%
\special{pa 2600 2200}%
\special{fp}%
%
\special{pn 8}%
\special{pa 2600 2200}%
\special{pa 2400 3600}%
\special{fp}%
\put(25.1000,-37.3000){\makebox(0,0){$u_{3k-3}$}}%
\put(35.1000,-37.3000){\makebox(0,0){$w_{3k-1}$}}%
\put(51.1000,-29.4000){\makebox(0,0){$w_{3k+1}$}}%
\put(57.1000,-23.3000){\makebox(0,0){$u_{3k+1}$}}%
\put(63.0000,-12.9000){\makebox(0,0){$w_{3k+2}$}}%
\put(35.0000,-12.9000){\makebox(0,0){$v_{3k-2}$}}%
\put(41.3000,-20.8000){\makebox(0,0){$v_{3k-1}$}}%
\put(35.0000,-29.1000){\makebox(0,0){$u_{3k-1}$}}%
\put(42.0000,-29.8000){\makebox(0,0){$w_{3k}$}}%
\put(49.5000,-20.9000){\makebox(0,0){$s = v_{3k}$}}%
\put(38.3000,-17.4000){\makebox(0,0){$e_2$}}%
\put(59.2000,-19.4000){\makebox(0,0){$e_1$}}%
%
\special{pn 8}%
\special{pa 4200 2800}%
\special{pa 4800 1400}%
\special{dt 0.045}%
\special{pa 4800 1400}%
\special{pa 4800 1400}%
\special{dt 0.045}%
%
\special{pn 8}%
\special{sh 0}%
\special{ar 4800 1400 50 50  0.0000000 6.2831853}%
%
\special{pn 8}%
\special{sh 0}%
\special{ar 2200 1400 50 50  0.0000000 6.2831853}%
%
\special{pn 8}%
\special{sh 0}%
\special{ar 3400 3600 50 50  0.0000000 6.2831853}%
%
\special{pn 8}%
\special{sh 0}%
\special{ar 5000 3600 50 50  0.0000000 6.2831853}%
%
\special{pn 8}%
\special{sh 0}%
\special{ar 8000 2200 50 50  0.0000000 6.2831853}%
%
\special{pn 8}%
\special{sh 0}%
\special{ar 7400 2800 50 50  0.0000000 6.2831853}%
%
\special{pn 8}%
\special{sh 0}%
\special{ar 7200 2200 50 50  0.0000000 6.2831853}%
%
\special{pn 8}%
\special{sh 0}%
\special{ar 6200 2200 50 50  0.0000000 6.2831853}%
%
\special{pn 8}%
\special{sh 0}%
\special{ar 5800 2800 50 50  0.0000000 6.2831853}%
%
\special{pn 8}%
\special{sh 0}%
\special{ar 6600 2800 50 50  0.0000000 6.2831853}%
%
\special{pn 8}%
\special{sh 0}%
\special{ar 5000 2800 50 50  0.0000000 6.2831853}%
%
\special{pn 8}%
\special{sh 0}%
\special{ar 4600 2200 50 50  0.0000000 6.2831853}%
%
\special{pn 8}%
\special{sh 0}%
\special{ar 4200 2800 50 50  0.0000000 6.2831853}%
%
\special{pn 8}%
\special{sh 0}%
\special{ar 3600 2800 50 50  0.0000000 6.2831853}%
%
\special{pn 8}%
\special{sh 0}%
\special{ar 3200 2800 50 50  0.0000000 6.2831853}%
%
\special{pn 8}%
\special{sh 0}%
\special{ar 2400 2800 50 50  0.0000000 6.2831853}%
%
\special{pn 8}%
\special{sh 0}%
\special{ar 1800 2800 50 50  0.0000000 6.2831853}%
%
\special{pn 8}%
\special{sh 0}%
\special{ar 2000 2200 50 50  0.0000000 6.2831853}%
%
\special{pn 8}%
\special{sh 0}%
\special{ar 2600 2200 50 50  0.0000000 6.2831853}%
%
\special{pn 8}%
\special{sh 0}%
\special{ar 3400 2200 50 50  0.0000000 6.2831853}%
%
\special{pn 8}%
\special{sh 0}%
\special{ar 4000 2200 50 50  0.0000000 6.2831853}%
%
\special{pn 8}%
\special{sh 0}%
\special{ar 5600 2200 50 50  0.0000000 6.2831853}%
\end{picture}%
\caption{The even kernel graphs with respect to $S_1$ and $S_2$ in Lemma~\ref{lem5}}\label{fig5}
\end{figure}
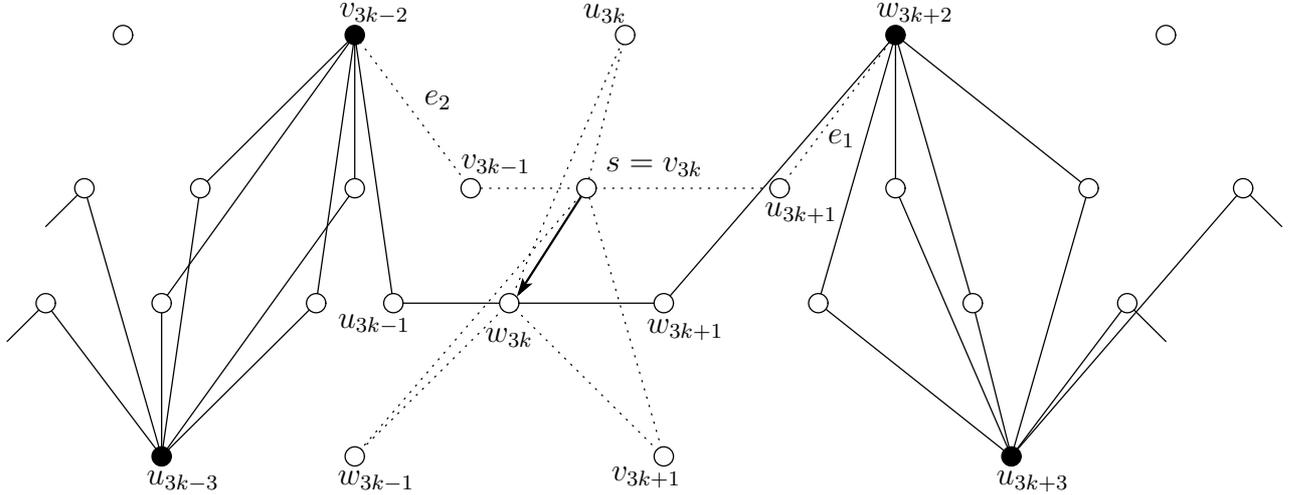

First, Alice moves the token to $w_{3k}$. Then Bob should move the token to $u_{3k-1}$ or $w_{3k+1}$, otherwise Bob will lose the game immediately. 
Consider the case $u_{3k-1}$ (resp., $w_{3k+1}$). Then Alice moves the token to $v_{3k-2}$ (resp., $w_{3k+2}$). 
By following the even kernel graph with respect to $S_1$ (resp., $S_2$), 
we see that Alice can move the token back to $v_{3k-2}$ (resp., $w_{3k+2}$) again. 
Finally, 
Bob must move the token to $v_{3k-1}$ (resp., $u_{3k+1}$) from $v_{3k-2}$ (resp., $w_{3k+2}$). 
Therefore, Alice can win the game. 
\end{proof}

We can now prove Theorem~\ref{thm:octapath}. 
\begin{proof}[Proof of Theorem~\ref{thm:octapath}]
All of the cases directly follow from Lemmas~\ref{lem1}--\ref{lem5} as follows
(where each ``$\equiv$'' means that ``congruent modulo $3$"):
\begin{center}
\begin{tabular}{c|c|c|c} 
             &$n \equiv 0$     &$n \equiv 1$     &$n \equiv 2$ \\ \hline
$p \equiv 0$ &Lemma~\ref{lem5} &Lemma~\ref{lem2} &Lemmas~\ref{lem1} $\&$ \ref{lem4} \\ \hline
$p \equiv 1$ &Lemmas~\ref{lem1} $\&$ \ref{lem4} &Lemma~\ref{lem2} &Lemma~\ref{lem3} \\ \hline
$p \equiv 2$ &Lemma~\ref{lem4} &Lemma~\ref{lem4} &Lemma~\ref{lem4} \\ 
\end{tabular}
\end{center}
\end{proof}

\section{Conclusion}

In this paper,
focusing on 3-chromatic Eulerian triangulations of surfaces,
we completely determine the winner of the feedback game on several classes of the graphs.
In general, for any surface $F^2$ which is not the sphere,
there are infinitely many non-3-colorable Eulerian triangulations of $F^2$.
Since Theorem~\ref{thm:deg4} strongly depends on the 3-colorability of the graphs,
it is not clear which player wins the game on 
a non-3-colorable Eulerian triangulation $G$ even if $c_{4,2}(G)=0$.
In fact, as shown in Figure~\ref{fig7},
there exists a 5-chromatic Eulerian triangulation $G$ of the projective plane
with $c_{4,2}(G)=0$
such that Alice wins the game on $G$ with a starting vertex $s$;
she first moves the token on $s$ to $v$,
and after that, since $s$ is adjacent to all other vertices,
Alice can move the token back to $s$ regardless of Bob's move.

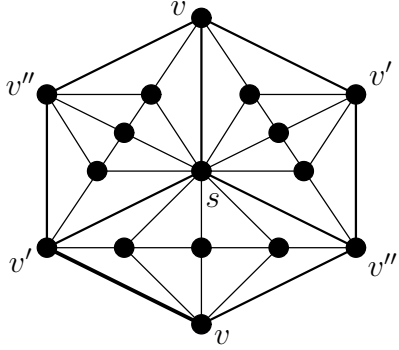
\begin{figure}[htb]
\centering
\unitlength 0.1in
\begin{picture}( 20.1000, 17.8600)( 24.4000,-24.5000)
%
\special{pn 8}%
\special{sh 1.000}%
\special{ar 3600 800 50 50  0.0000000 6.2831853}%
%
\special{pn 8}%
\special{sh 1.000}%
\special{ar 4400 2000 50 50  0.0000000 6.2831853}%
%
\special{pn 8}%
\special{sh 1.000}%
\special{ar 2800 2000 50 50  0.0000000 6.2831853}%
%
\special{pn 8}%
\special{sh 1.000}%
\special{ar 2800 1200 50 50  0.0000000 6.2831853}%
%
\special{pn 8}%
\special{sh 1.000}%
\special{ar 3600 2400 50 50  0.0000000 6.2831853}%
%
\special{pn 8}%
\special{sh 1.000}%
\special{ar 4400 1200 50 50  0.0000000 6.2831853}%
%
\special{pn 8}%
\special{sh 1.000}%
\special{ar 3600 1600 50 50  0.0000000 6.2831853}%
%
\special{pn 13}%
\special{pa 3600 1600}%
\special{pa 3600 800}%
\special{fp}%
%
\special{pn 13}%
\special{pa 3600 800}%
\special{pa 2800 1200}%
\special{fp}%
%
\special{pn 13}%
\special{pa 2800 1200}%
\special{pa 2800 2000}%
\special{fp}%
%
\special{pn 20}%
\special{pa 2800 2000}%
\special{pa 3600 2400}%
\special{fp}%
%
\special{pn 13}%
\special{pa 3600 2400}%
\special{pa 4400 2000}%
\special{fp}%
%
\special{pn 13}%
\special{pa 4400 2000}%
\special{pa 4400 1200}%
\special{fp}%
%
\special{pn 13}%
\special{pa 4400 1200}%
\special{pa 3600 800}%
\special{fp}%
%
\special{pn 13}%
\special{pa 3600 1600}%
\special{pa 2800 2000}%
\special{fp}%
%
\special{pn 13}%
\special{pa 3600 1600}%
\special{pa 4400 2000}%
\special{fp}%
%
\special{pn 8}%
\special{pa 4400 2000}%
\special{pa 3600 800}%
\special{fp}%
%
\special{pn 8}%
\special{pa 3600 800}%
\special{pa 2800 2000}%
\special{fp}%
%
\special{pn 8}%
\special{pa 2800 2000}%
\special{pa 4400 2000}%
\special{fp}%
%
\special{pn 8}%
\special{pa 4000 2000}%
\special{pa 3600 1600}%
\special{fp}%
%
\special{pn 8}%
\special{pa 3600 1600}%
\special{pa 3200 2000}%
\special{fp}%
%
\special{pn 8}%
\special{pa 3200 2000}%
\special{pa 3600 2400}%
\special{fp}%
%
\special{pn 8}%
\special{pa 3600 2400}%
\special{pa 4000 2000}%
\special{fp}%
%
\special{pn 8}%
\special{pa 3600 1600}%
\special{pa 3600 2400}%
\special{fp}%
%
\special{pn 8}%
\special{pa 3600 1600}%
\special{pa 2800 1200}%
\special{fp}%
%
\special{pn 8}%
\special{pa 3600 1600}%
\special{pa 4400 1200}%
\special{fp}%
%
\special{pn 8}%
\special{pa 3860 1200}%
\special{pa 4400 1200}%
\special{fp}%
%
\special{pn 8}%
\special{pa 3600 1600}%
\special{pa 3850 1200}%
\special{fp}%
%
\special{pn 8}%
\special{pa 3340 1200}%
\special{pa 3600 1600}%
\special{fp}%
%
\special{pn 8}%
\special{pa 3310 1200}%
\special{pa 2810 1200}%
\special{fp}%
%
\special{pn 8}%
\special{pa 2800 1200}%
\special{pa 3050 1600}%
\special{fp}%
%
\special{pn 8}%
\special{pa 3050 1600}%
\special{pa 3600 1600}%
\special{fp}%
%
\special{pn 8}%
\special{sh 1.000}%
\special{ar 4000 2000 50 50  0.0000000 6.2831853}%
%
\special{pn 8}%
\special{sh 1.000}%
\special{ar 3600 2000 50 50  0.0000000 6.2831853}%
%
\special{pn 8}%
\special{sh 1.000}%
\special{ar 3200 2000 50 50  0.0000000 6.2831853}%
%
\special{pn 8}%
\special{sh 1.000}%
\special{ar 3200 1400 50 50  0.0000000 6.2831853}%
%
\special{pn 8}%
\special{sh 1.000}%
\special{ar 4000 1400 50 50  0.0000000 6.2831853}%
%
\special{pn 8}%
\special{sh 1.000}%
\special{ar 3850 1200 50 50  0.0000000 6.2831853}%
%
\special{pn 8}%
\special{sh 1.000}%
\special{ar 3340 1200 50 50  0.0000000 6.2831853}%
%
\special{pn 8}%
\special{sh 1.000}%
\special{ar 3060 1600 50 50  0.0000000 6.2831853}%
%
\special{pn 8}%
\special{pa 3600 1600}%
\special{pa 4140 1600}%
\special{fp}%
%
\special{pn 8}%
\special{pa 4140 1600}%
\special{pa 4400 1200}%
\special{fp}%
%
\special{pn 8}%
\special{sh 1.000}%
\special{ar 4130 1600 50 50  0.0000000 6.2831853}%
\put(36.6000,-17.5000){\makebox(0,0){$s$}}%
\put(34.8200,-7.4900){\makebox(0,0){$v$}}%
\put(37.1000,-24.7000){\makebox(0,0){$v$}}%
\put(26.6500,-20.6500){\makebox(0,0){$v'$}}%
\put(45.3000,-10.9000){\makebox(0,0){$v'$}}%
\put(26.6500,-11.4500){\makebox(0,0){$v''$}}%
\put(45.4000,-20.9000){\makebox(0,0){$v''$}}%
\end{picture}%
\caption{The 5-chromatic Eulerian triangulation $G$ of the projective plane
with $c_{4,2}(G)=0$}
\label{fig7}
\end{figure}

On the other hand, 
the proof of Theorem~\ref{thm:c_{4,2}=m} does not strongly depend on the 3-colorability
so much,
since we can obtain a similar statement as Theorem~\ref{thm:c_{4,2}=m}
by showing that for a fixed surface $F^2$,
there exists a non-3-colorable Eulerian triangulation $G$ of $F^2$ with $c_{4,2}(G)=0$.
Thus we guess that 
Theorem~\ref{thm:c_{4,2}=m} also holds for non-3-colorable Eulerian triangulations.

As in the octahedral path,
determining the winner of the game on a 3-chromatic Eulerian triangulation $G$ 
with $c_{4,2}(G) > 0$ is not easy,
since a choice of the starting vertex changes the winner of the game 
(as Theorem~\ref{thm:octapath}).
Similarly to the octahedral path,
a 6-regular triangulation seems to be 
the next reasonable concrete class of Eulerian triangulations,
since only the torus and Klein bottle admit such triangulations 
and they are completely classified (see~\cite{A,KMNNST,N}).
Therefore, we conclude this paper with proposing the following problem.

\begin{prob}
Completely determine the winner of the feedback game on
$3$-chromatic $6$-regular triangulations of the torus or Klein bottle.
\end{prob}

\subsection*{Acknowledgement}
The first author is supported by JSPS Grant-in-Aid for Young Scientists (B) 17K14177.
The third author is supported by JSPS 
Grant-in-Aid for Early-Career Scientists 19K14583.

\end{document}